\documentclass[a4paper,11pt]{article}
\usepackage{bm,mathrsfs,amsmath,amsthm,amssymb,eepic,amscd,longtable,array,graphicx}
\usepackage[dvips]{color}
\newcommand{\nc}{\newcommand}
\nc{\tcr}{\textcolor{red}}
\nc{\rnc}{\renewcommand}
\nc{\nn}{\nonumber}
\nc{\der}{{\partial}}
\rnc{\Im}{{\rm{Im}\,}}
\rnc{\Re}{{\rm{Re}\,}}
\nc{\db}{\displaybreak[0]\\}
\nc{\bra}{\langle}
\nc{\ket}{\rangle}
\nc{\bs}{\boldsymbol}

\DeclareMathOperator{\Tr}{Tr}

\DeclareMathOperator{\End}{End}

\newtheorem{theorem}{Theorem}[section]
\newtheorem{lemma}[theorem]{Lemma}

\newtheorem{proposition}[theorem]{Proposition}
\newtheorem{corollary}[theorem]{Corollary}

\theoremstyle{definition}
\newtheorem{definition}[theorem]{Definition}
\newtheorem{example}[theorem]{Example}

\numberwithin{equation}{section}

\numberwithin{equation}{section}

\begin{document}%
%

\title{Quantum integrable combinatorics of Schur polynomials}

\author{
Kohei Motegi$^1$\thanks{E-mail: motegi@gokutan.c.u-tokyo.ac.jp} \,
and
Kazumitsu Sakai$^2$\thanks{E-mail: sakai@gokutan.c.u-tokyo.ac.jp}
\\\\
$^1${\it Faculty of Marine Technology, Tokyo University of Marine Science and Technology,}\\
 {\it Etchujima 2-1-6, Koto-Ku, Tokyo, 135-8533, Japan} \\
\\
$^2${\it Institute of physics, University of Tokyo,} \\ 
{\it Komaba 3-8-1, Meguro-ku, Tokyo 153-8902, Japan}
\\\\
\\
}

\date{\today}

\maketitle

\begin{abstract}
We examine and present new combinatorics for the Schur polynomials
from the viewpoint of quantum integrability.
We introduce and analyze an integrable six-vertex model
which can be viewed as a certain degeneration model from
a $t$-deformed boson model.
By a detailed analysis of the wavefunction from the quantum inverse scattering
method,
we present a novel combinatorial formula which expresses the Schur polynomials
by using an additional parameter, which is in the same sense but
different from the Tokuyama formula.
We also give an algebraic analytic proof for the Cauchy identity
and make applications of the domain wall boundary partition functions
to the enumeration of alternating sign matrices.
\end{abstract}

\section{Introduction}
Schur polynomials is an ubiquitous object in
mathematics and mathematical physics,
ranging from representation theory, combinatorics, enumerative geometry
to knot theory.
Being one of the most fundamental symmetric polynomials,
extensive studies and numerous combinatorial identities have been found for the
Schur polynomials.
Many mathematical objects are introduced for which the
Schur polynomials serves as a building block since
one can use various properties found for the Schur polynomials.
One typical example is a stochastic process called
the Schur process \cite{OR} whose
probability measure is given by the Schur polynomials,
from its property many useful combinatorial formulae
can be employed to study the dynamics of the process.

The Schur polynomials itself has many expressions.
One of the most famous ones is the Jacobi-Trudi identity
which expresses Schur polynomials in terms of elementary symmetric polynomials.
Besides these traditional identities,
a very interesting formula called the Tokuyama formula \cite{To,Ok,HK}
was found which expresses Schur polynomials in a combinatorial form.
The distinctive feature of the formula is that the expression
is given in terms not only of the symmetric variables
but also with an additional parameter which does not appear in the original
definition of the Schur polynomials.
Specializing the additional parameter, the 
Tokuyama formula reduces to the determinant formula,
Weyl character formula and so on.
In this sense, the Tokuyama formula can be viewed as a deformation
of Weyl character.

Recently, the Tokuyama formula has found its interpretation in terms of
integrable vertex models \cite{BBF}.
The idea to understand this Tokuyama formula from the point of view of
quantum integrability
is to extend the $L$-operator for the free-fermion model (the corresponding $R$-matrix is the trigonometric Felderhof model) to
an $L$-operator including an additional parameter
besides the spectral parameter. The newly-introduced parameter plays the role
of refining the combinatorial expression for the Schur polynomials.
This idea of relating Tokuyama formula which is a 
deformation of Weyl character formula to integrable vertex models
have been recently generalized to the factorial Schur polynomials \cite{BMN}
and other types of symmetric polynomials \cite{Tab,BBCG,BS,HK2}
by introducing inhomogeneous parameters or
changing boundary conditions.

In this paper, we present another type of combinatorial formula for the
Schur polynomials with an additional parameter
by investigating a partition function of an integrable six-vertex model.
We find the following combinatorial formula for the Schur polynomials
\begin{align}
s_\lambda(\bs{z})=&
\frac{1}{\prod_{1 \le j < k \le N}(z_j+z_k+2 \beta z_j z_k)}
\sum_{x^{(N)} \succ x^{(N-1)} \succ \dots \succ x^{(0)}=\phi} \prod_{k=1}^N
\Bigg\{z_k^{\sum_{j=1}^k x_j^{(k)}-\sum_{j=1}^{k-1} x_j^{(k-1)}-1} \nonumber \\
\times&\Bigg( \frac{2(1+\beta z_k)}{1+2 \beta z_k} \Bigg)^{\#(x^{(k)}|x^{(k-1)})-1}
\prod_{j=1}^{k-1}
\Bigg(
1+2 \beta z_k(1-\delta_{x_j^{(k-1)} x_{j+1}^{(k)}})
\Bigg)
\Bigg\}, \label{mainformula}
\end{align}
where $\beta$ is an arbitrary parameter.
$x^{(k)}=(x_1^{(k)},\dots,x_k^{(k)})$, $k=1,\dots,N$ are strict partitions
satisfying the interlacing relations
$x^{(N)} \succ x^{(N-1)} \succ \dots \succ x^{(0)}=\phi$,
$x^{(N)}$ is fixed by the Young diagram
$\lambda=(\lambda_1,\dots,\lambda_N)$ as
$\lambda_j=x_j^{(N)}-N+j-1$, and $\#(y|x)$ denotes
the number of parts in $y$ which are not in $x$.

We obtain this combinatorial formula from an integrable six-vertex model.
The six-vertex model we consider can be regarded as a certain
degeneration of a $t$-deformed non-Hermitian boson model introduced in
\cite{SW}. The corresponding wavefunction seems to be analyzed by the
coordinate Bethe ansatz \cite{Ta} which at the degeneration point
is given by the Schur polynomials times a factor including
an additional parameter.
To obtain the combinatorial formula \eqref{mainformula},
we investigate the corresponding model from the point of view of the
quantum inverse scattering method, i.e., starting from the $L$-operator
which is the most fundamental object in quantum integrable models.
We view the wavefunction as a partition function
of a six-vertex model and evaluate it in two ways.
First, we evaluate the partition function directly to show that
it is given by a Schur polynomials times an additional factor.
Another way of evaluation is to take the viewpoint that
the partition function consists of a layer of $B$-operators,
and calculate the matrix elements of a single $B$-operator
to show a summation formula for the partition function.
Combining the two expressions give the combinatorial formula
\eqref{mainformula}.

Besides the combinatorial formula,
we examine the scalar products and present an algebraic analytic
proof for its determinant form,
which combined with the wavefunction gives us the celebrated Cauchy formula
for the Schur polynomials.
We also make an application of the domain wall boundary partition function
to show a simple formula for a particular generating function of
alternating sign matrices.

This paper is organized as follows.
We introduce the integrable six-vertex model as a particular
degeneration of a non-Hermitian $t$-deformed boson model in the next section.
We give an algebraic analytic proof for the scalar product in section 3.
In sections 4 and 5, we evaluate the matrix elements and
the wavefunction of the six-vertex model.
We combine the results of sections 3, 4 and 5 to present
the combinatorial formulae for the Schur polynomials in section 6.
In section 7, we give an application of the domain wall boundary
partition function to a generating function of the alternating sign matrices.

\section{Non-Hermitian $t$-deformed boson
model and reduction to the six-vertex model}

We introduce the $t$-deformed non-Hermitian boson model \cite{SW,Ta}
and its equivalent six-vertex model at the point $t=-1$.
The model is characterized by
$t$-deformed boson algebra
with generators $N,B,B^\dagger$ satisfying the
following relations
\begin{align}
[B, B^\dagger]=t^{N}(1-t), \quad
[N, B]=-B, \quad
[N, B^\dagger]=B^\dagger.
\end{align}
These generators act on a Fock space $\mathcal{F}$ spanned by
the orthonormal basis $| n \rangle \ (n\ge 0)$ as
\begin{align}
B|n \rangle=|n-1 \rangle, \
B^\dagger |n \rangle=(1-t^{n+1})|n+1 \rangle, \
N|n \rangle= n|n \rangle.
\end{align}
The operators 
the dual orthonormal basis $\langle n|$ ($n\ge 0$)
satisfying $\langle n|m \rangle=\delta_{nm}$ as
\begin{align}
\langle n| B^\dagger=(1-t^n) \langle n-1|, \
\langle n| B=\langle n+1|, \
\langle n| N=n \langle n|.
\end{align}

We consider the following $L$-operator which is a slightly
deformation of the one in \cite{SW}
\begin{align}
\mathcal{L}_{a j}(v)=
\begin{pmatrix}
1-\beta v t^{N_j} & v B_j^\dagger  \\
B_j & v
\end{pmatrix}_a,
\label{Lop-boson}
\end{align}
acting on the tensor product $W_a \otimes \mathcal{F}_j$ of the
complex two-dimensional space $W_a$ and the Fock space at the
$j$th site $\mathcal{F}_j$.
Let us denote the orthonormal basis of $W_a$ and its dual as
$\{|0 \rangle_a, |1 \rangle_a \}$ and $\{{}_a \langle 0|, {}_a \langle 1|\}$.
The explicit forms of the matrix elements of  \eqref{Lop-boson}
in the orthonormal basis are
\begin{align}
{}_a \langle 0| {}_j \langle m_j | \mathcal{L}_{a j}(v)
|0 \rangle_a | n_j \rangle_j
&=(1-\beta v t^{n_j}) \delta_{m_j, n_j}, \label{Lop1} \\
{}_a \langle 1| {}_j \langle m_j | \mathcal{L}_{a j}(v)
|0 \rangle_a | n_j \rangle_j
&=\delta_{m_j, n_j-1}, \\
{}_a \langle 0| {}_j \langle m_j | \mathcal{L}_{a j}(v)
|1 \rangle_a | n_j \rangle_j
&=v(1-t^{n_j+1}) \delta_{m_j, n_j+1}, \label{specialmatrixelement} \\
{}_a \langle 1| {}_j \langle m_j | \mathcal{L}_{a j}(v)
|1 \rangle_a | n_j \rangle_j
&=v \delta_{m_j, n_j} \label{Lop2}.
\end{align}
In figure~\ref{boson-L}, we depict the non-zero elements
of the $L$-operator.

\begin{figure}[tt]
\begin{center}
\includegraphics[width=0.7\textwidth]{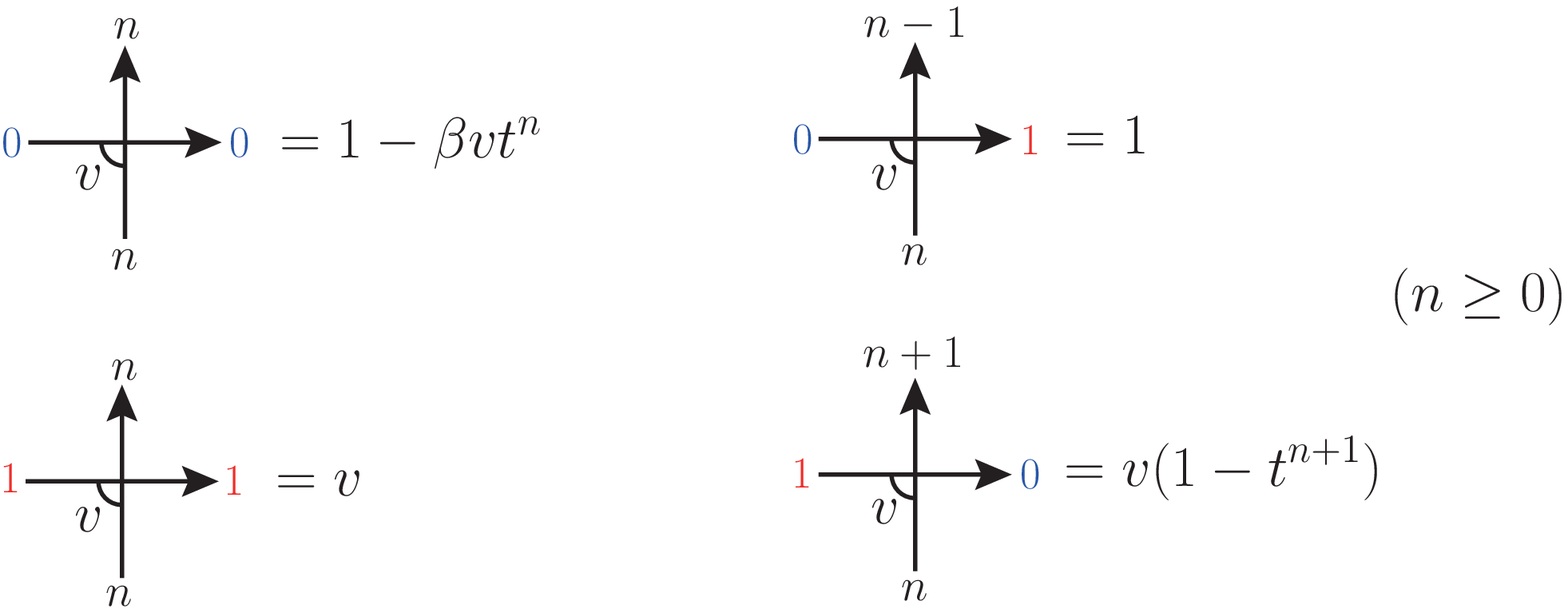}
\end{center}
\caption{The non-zero elements of the $L$-operator.
}
\label{boson-L}
\end{figure}

The $L$-operator satisfies the intertwining relation ($RLL$-relation)
\begin{align}
R_{ab}(u,v)\mathcal{L}_{a j}(u)\mathcal{L}_{b j}(v)=
\mathcal{L}_{b j}(v)\mathcal{L}_{a j}(u)R_{ab}(u,v),
\label{RLL-boson}
\end{align}
which acts on $W_a \otimes W_b \otimes \mathcal{F}_j$.
The $R$-matrix
\begin{align}
R(u,v)
=
\begin{pmatrix}
u-tv & 0 & 0 & 0 \\
0 & t(u-v) & (1-t)u & 0 \\
0 & (1-t)v & u-v & 0 \\
0 & 0 & 0 & u-tv
\end{pmatrix}, 
\label{Rmatrix}
\end{align}
is a solution to the Yang-Baxter equation:
\begin{align}
R_{ab}(u,v)R_{ac}(u,w)R_{bc}(v,w)=
R_{bc}(v,w)R_{ac}(u,w)R_{ab}(u,v).
\label{YBE}
\end{align}

From the $L$-operator, we construct the monodromy matrix
\begin{align}
\mathcal{T}_{a}(v)=\prod_{j=1}^{M} \mathcal{L}_{a j}(v),
\label{monodromy1}
\end{align}
which acts on $W_a \otimes (\mathcal{F}_1\otimes\dots\otimes 
\mathcal{F}_{M})$. Tracing
out the auxiliary space, one defines the transfer matrix 
$\tau(v)\in \mathrm{End} (\mathcal{F}^{\otimes M})$:
\begin{align}
\tau (v)=\Tr_{W_a} \mathcal{T}_{a}(v). 
\label{TM}
\end{align}

Simplifications happen for the $t$-deformed boson model
at $t=-1$.
Due to the matrix element \eqref{specialmatrixelement},
when we construct $N$-particle states from the vacuum,
each site can only be occupied by at most one boson.
This is equivalent to the reduction of the $t$-deformed boson model
to the six-vertex model whose $L$-operator is given by
\begin{align}
\mathcal{L}_{aj}(v)
=
\begin{pmatrix}
1-\beta v & 0 & 0 & 0 \\
0 & 1+\beta v & 2v & 0 \\
0 & 1 & v & 0 \\
0 & 0 & 0 & v
\end{pmatrix}_{aj}.
\label{sixvertexLoperator}
\end{align}
One can check that 
\eqref{sixvertexLoperator} indeed satisfies the intertwining relation
\eqref{RLL-boson}.
We analyze the structure of this six-vertex model in this paper.
We remark that at the other degenerated point $t=0$ of the $t$-deformed boson
model, the model is called the non-Hermitian phase model \cite{BN},
whose wavefunction \cite{MS} constructed from the $L$-operator
is given by the Grothendieck polynomials \cite{LS,FK,Bu,IN,KN},
which furthermore
reduces to the Schur polynomials by taking $\beta=0$ where the corresponding
model is the {\it Hermitian} phase model \cite{Bo,SU,KS,Wh}.
On the other hand, the wavefunction at $\beta=0$ for generic $t$
is given by the Hall-Littlewood polynomials \cite{Ts,Kor}.

\section{Scalar Products of state vectors of the six-vertex model}
In this section,
we first construct a state vector of the
integrable $t$-deformed boson model
following the standard procedure of the
quantum inverse scattering method
(i.e. the algebraic Bethe ansatz) which
is based on the Yang-Baxter algebra.
The $N$-particle state $|\Psi(\{ v \}_N) \rangle$
is characterized by $N$ unknown numbers
$v_j\in \mathbb{C}$ ($1 \le j \le N$),
and is not an eigenvector of the transfer matrix in general.
However, if the parameters $v_j$ satisfy a set of constraints
called the Bethe ansatz equation, the $N$-particle state
becomes an eigenstate.
We call the state $|\Psi(\{ v \}_N) \rangle$ on-shell state
if $\{ v \}$ satisfies the Bethe ansatz equation,
and off-shell state if no constraints are imposed on $\{ v \}$.
We then restrict the analysis to $t=-1$,
which is equivalent to considering the six-vertex model
\eqref{sixvertexLoperator},
and prove the determinant form for the
scalar products $\langle \Psi(\{u\}_N)|\Psi(\{v\}_N) \rangle$
which is the inner product between
the off-shell state $|\Psi(\{v\}_N) \rangle$
and the dual off-shell state $\langle \Psi(\{u\}_N)|$.

From the $L$-operator, we construct the monodromy matrix
\begin{align}
\mathcal{T}_{a}(v,\{w\})=\prod_{j=1}^{M} \mathcal{L}_{a j}(v/w_{M+1-j})
=
\begin{pmatrix}
\mathcal{A}(v,\{ w \}) & \mathcal{B}(v, \{ w \})  \\
\mathcal{C}(v,\{ w \}) & \mathcal{D}(v, \{ w \})
\end{pmatrix}_{a}, 
\label{monodromy}
\end{align}
which acts on $W_a \otimes (\mathcal{F}_1\otimes\dots\otimes 
\mathcal{F}_{M})$.  Here we introduced the inhomogeneous parameters 
$w_1,\dots,w_M \in \mathbb{C}$. See figure~\ref{abcdpic} for
the graphical description of the elements of the monodromy
matrix.
\begin{figure}[ttt]
\begin{center}
\includegraphics[width=0.9\textwidth]{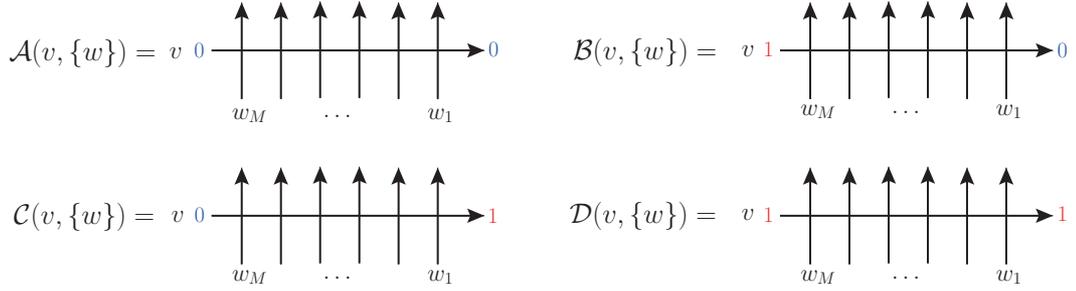}
\end{center}
\caption{The diagrammatic representation of 
the elements of the monodromy matrix \eqref{monodromy}
with inhomogeneous parameters $w_{M+1-j}$ associated
with site $j$.}
\label{abcdpic}
\end{figure}

Taking the homogeneous limit 
$w_j\to 1$ ($1\le j\le M$), \eqref{monodromy1} is recovered:
\begin{align}
\mathcal{T}_a(v,\{w\})|_{w_1=1,\dots,w_M=1}=\mathcal{T}_a(v).
\label{homogeneous}
\end{align}
As in the above equation, hereafter we will omit $\{w\}$ for 
the quantities in the homogeneous limit (e.g. 
$\mathcal{A}(v):=\mathcal{A}(v,\{w\})|_{w_1=1,\dots,w_M=1}$).
Tracing
out the auxiliary space, one defines the transfer matrix 
$\tau(v,\{w\})\in \mathrm{End} (\mathcal{F}^{\otimes M})$:
\begin{align}
\tau (v,\{w\})=\Tr_{W_a} \mathcal{T}_{a}(v,\{w\}). 
\label{TM2}
\end{align}

The repeated applications of the $RLL$-relation leads to the
intertwining relation
\begin{align}
R_{ab}(u,v)\mathcal{T}_{a}(u, \{ w \})\mathcal{T}_{b}(v, \{ w \})=
\mathcal{T}_{b}(v, \{ w \})\mathcal{T}_{a}(u, \{ w \})R_{a b}(u,v) \label{RTT}.
\end{align}
Some of the elements of the intertwining relation are
\begin{align}
&\mathcal{C}(u, \{ w \})\mathcal{B}(v, \{ w \})
-t \mathcal{B}(v, \{ w \})\mathcal{C}(u, \{ w \}) \nonumber \\
&=g(u,v)(\mathcal{A}(v, \{ w \})\mathcal{D}(u, \{ w \})
-\mathcal{A}(u, \{ w \})\mathcal{D}(v, \{ w \})), \\
&\mathcal{A}(u, \{ w \})\mathcal{B}(v, \{ w \})
=f(u,v)\mathcal{B}(v, \{ w \})\mathcal{A}(u, \{ w \})
+g(u,v)\mathcal{B}(u, \{ w \})\mathcal{A}(v, \{ w \}), \\
&\mathcal{D}(u, \{ w \})\mathcal{B}(v, \{ w \})
=f(v,u)\mathcal{B}(v, \{ w \})\mathcal{D}(u, \{ w \})
-g(u,v)\mathcal{B}(u, \{ w \})\mathcal{D}(u, \{ w \}), \\
&{[} \mathcal{B}(u,\{w\}),\mathcal{B}(v,\{w\}){]}
={[} \mathcal{C}(u,\{w\}),\mathcal{C}(v,\{w\}){]}=0,
\end{align}
where
\begin{align}
f(u,v)=\frac{ut-v}{u-v}, \quad  g(u,v)=\frac{(1-t)v}{u-v}.
\end{align}
The transfer matrix $\tau(v,\{w\})$ is then expressed as elements
of the monodromy matrix:
\begin{align}
\tau(v,\{w\})=\Tr_{W_a} \mathcal{T}_{a}(v,\{w\})
=\mathcal{A}(v,\{w\})+\mathcal{D}(v,\{w\}).
\label{transfer}
\end{align}

The arbitrary $N$-particle state $|\Psi(\{v \}_N) \ket$
(resp. its dual $\bra \Psi(\{v \}_N)|$) 
(not normalized) with $N$ spectral parameters
$\{ v \}_N=\{ v_1,\dots,v_N \}$
is constructed by a multiple action
of $\mathcal{B}$ (resp. $\mathcal{C}$) operator on the vacuum state 
$|\Omega \ket:= | 0^{M} \ket:=|0\ket_1\
\otimes \dots \otimes |0\ket_{M}$
(resp. $\bra \Omega| := \bra 0^{M}|:=
{}_1\bra 0|\otimes\dots \otimes{}_{M}\bra 0|$):
\begin{align}
|\Psi(\{v \}_N,\{ w \}) \ket=\prod_{j=1}^N \mathcal{B}(v_j,\{ w \})| \Omega \ket,
\quad
\bra \Psi(\{v \}_N,\{ w \})|=\bra \Omega| \prod_{j=1}^N
\mathcal{C}(v_j,\{ w \}).
\label{statevector-B}
\end{align}
By the standard procedure of the algebraic Bethe ansatz,
we have the followings.
\begin{proposition}\label{eigenstate}
The $N$-particle state $|\Psi(\{v \}_N,\{w\}) \ket$
and its dual $\langle \Psi(\{v \}_N,\{w\})|$ become
an eigenstate (on-shell states) of the transfer matrix 
\eqref{transfer} when the set of parameters $\{v\}_N$
satisfies the Bethe ansatz equation:
\begin{align}
\frac{a(v_j,\{w\})}{d(v_j,\{w\})}=-\prod_{k=1}^N\frac{f(v_k,v_j)}{f(v_j,v_k)}
\label{BAE},
\end{align}
where
\begin{align}
a(v,\{w\})=\prod_{j=1}^M \left(1-\frac{\beta v}{w_j} \right), \quad
d(v,\{w\})=\prod_{j=1}^M \frac{v}{w_j}.
\end{align}
Then the eigenvalue of the transfer matrix is given by
\begin{align}
\tau(v,\{w\})=a(v,\{w\})\prod_{j=1}^N f(v,v_j)
             +d(v,\{w\})\prod_{j=1}^N f(v_j,v).
\label{EV}
\end{align}
\end{proposition}
For the case $t=-1$ which we investigate extensively in this paper,
the Bethe ansatz equation \eqref{BAE} is
\begin{align}
\prod_{k=1}^M \left( \frac{w_k}{v_j}-\beta \right)=(-1)^{N+1}, \ j=1,\dots,N.
\end{align}
This is the Bethe ansatz equation for free fermions
in the homogeneous limit $w_j\to 1$ ($1\le j \le N$).
In the analysis below, we do not impose these constraints between the
spectral parameters $\{ v \}$ and inhomogeneous parameters $\{ w \}$.
We remark that from this consideration on the Bethe ansatz equation,
one can imagine the wavefunction is given by the Schur polynomials.
The amazing thing is that a detailed analysis of its realization
as partition functions lead us to a new combinatorial formula
for the Schur polynomials itself.

The scalar product between the arbitrary off-shell state vectors, 
which is mainly considered in this section, is defined as
\begin{align}
\langle \Psi(\{ u \}_N,\{w\})| \Psi(\{ v \}_N,\{w\}) \rangle 
&=\langle \Omega| \prod_{j=1}^N \mathcal{C}(u_j,\{w\})
\prod_{k=1}^N \mathcal{B}(v_k,\{w\})| \Omega \rangle
\label{SP}
\end{align}
with $u_j, v_k\in \mathbb{C}$.

From now on, we specialize the parameter $t$
of the $t$-deformed boson algebra
to $t=-1$.
This is equivalent to considering the
six-vertex model \eqref{sixvertexLoperator}.
The following
determinant formula
in the homogeneous limit
$w_j\to 1$ ($1\le j \le N$) is valid.

\begin{theorem}{\label{scalarthm}}
The scalar product \eqref{SP} in the homogeneous limit 
$w_j\to1$ ($1\le j \le N$) is given by a determinant form:
\begin{align}
\bra \Psi(\{ u \}_N)| \Psi(\{ v \}_N) \ket
=\prod_{j=1}^N (2v_j)
\prod_{1 \le j<k \le N }
\frac{(u_j+u_k)(v_j+v_k)}{(u_j-u_k)(v_k-v_j)}
\mathrm{det}_N Q(\{ u \}_N|\{ v \}_N),
\label{generalscalar}
\end{align}
where $\{u\}_N$ and $\{v\}_N$ are arbitrary sets
of complex values (i.e. off-shell conditions), and
$Q$ 
is an $N \times N$ matrix with matrix elements
\begin{align}
Q(\{ u \}_N|\{ v \}_N)_{jk}=\frac{a(u_j)d(v_k)-d(u_j)a(v_k)}{v_k-u_j}
.
\label{element}
\end{align}
\end{theorem}
Here we will show the above determinant formula by
using a method initiated by Izergin-Korepin \cite{Kore,Ize}
for the domain wall boundary partition function
and recently developed by Wheeler \cite{Wheeler2}
in the calculation of the scalar product of the
spin-1/2 XXZ chain.
This procedure was applied to a family of integrable five-vertex model
\cite{MSvertex} which, in contrast to the spin-1/2 XXZ chain,
there was no need to impose the Bethe ansatz equation
(i.e. on-shell condition) to show the determinant formula.
For the six-vertex model we consider in this paper,
we also do not have to impose the Bethe ansatz equation, i.e.,
the determinant formula \eqref{generalscalar} is valid for arbitrary 
off-shell states.

First let us introduce the
following intermediate scalar products
which plays the key role for the proof
\begin{align}
S(\{ u \}_n | \{ v \}_N| \{ w \})=
\langle 
1^{N-n} 0^{M-N+n}| \prod_{j=1}^n \mathcal{C}(u_j,\{w\})
                                  \prod_{k=1}^N \mathcal{B}(v_k,\{w\}) 
 |\Omega 
\rangle. 
\label{intermediatedef}
\end{align}
See also figure~\ref{intermediatepic}.
The term ``intermediate" stems from the fact that \eqref{intermediatedef}
interpolates the scalar product ($n=N$) \eqref{SP} and the domain wall boundary 
partition function ($n=0$) \eqref{domainwallboundary} 
 (see also figure~\ref{initialpic}).
\begin{figure}[tt]
\begin{center}
\includegraphics[width=0.6\textwidth]{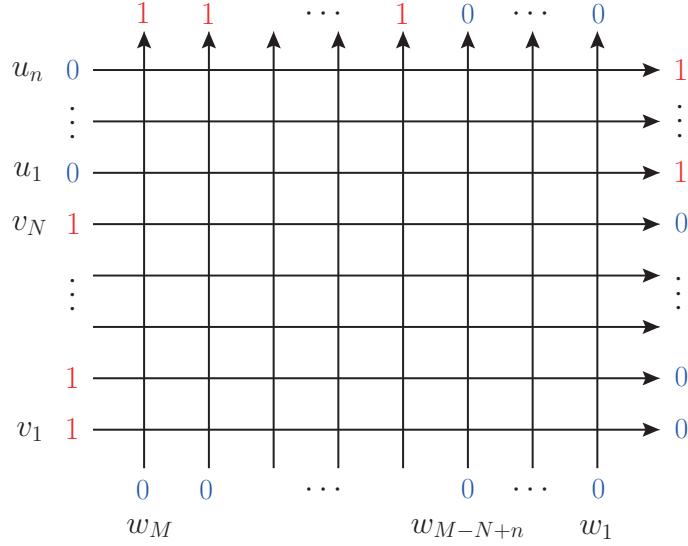}
\end{center}
\caption{The graphical representation of the intermediate scalar products 
\eqref{intermediatedef}
with inhomogeneous parameters $w_{M+1-j}$ associated
with site $j$.}
\label{intermediatepic}
\end{figure}
We have the following lemma regarding the properties of the intermediate 
scalar product.
\begin{lemma}{\label{property}}
The intermediate scalar product \eqref{intermediatedef} 
$S(\{ u \}_n | \{ v \}_N| \{ w \})$ satisfies the
following properties. 
\begin{enumerate}
\item $S(\{ u \}_n | \{ v \}_N| \{ w \})$ is symmetric with respect to
the variables $\{ w_1,\dots,w_{M-N+n} \}$. 
\item $\prod_{j=M-N+n+1}^M (1+\beta u_n/w_j)^{-1}
S(\{ u \}_n | \{ v \}_N| \{ w \})$ 
is a polynomial of degree $M-N+n-1$ in $u_n$. 
\item
The following recursive relations between the intermediate scalar products
hold
\begin{align}
&S(\{ u \}_n|\{ v \}_N|\{ w \})|_{u_n=\beta^{-1} w_{M-N+n}} \nn \\
& \qquad \qquad 
=\prod_{j=1}^{M-N+n-1} \frac{w_{M-N+n}}{\beta w_j}
\prod_{j=M-N+n+1}^M  \Bigg(1+\frac{w_{M-N+n}}{w_j}\Bigg)
S(\{ u \}_{n-1}|\{ v \}_N|\{ w \}). \label{recursiverelation}
\end{align}
\item  The case $n=0$ of the intermediate scalar products has the 
following form:
\begin{align}
S(\{ u \}_0|\{ v \}_N|\{ w \})
=\prod_{j=1}^N \frac{2v_j}{w_{M-N+j}^{j}}
\prod_{1 \le j<k \le N}(v_j+v_k)
\prod_{j=1}^N \prod_{k=1}^{M-N} \Bigg(1-\frac{\beta v_j}{w_k} \Bigg).
\label{initial}
\end{align}
\end{enumerate}
\end{lemma}

\begin{proof}
Property 1 follows from the $RLL$-relation
\begin{align}
&\widetilde{R}_{jk}(w_{M+1-j}/w_{M+1-k})
L_{a k}(u/w_{M+1-k})L_{a j}(u/w_{M+1-j}) \nonumber \\
&\qquad\qquad =L_{a j}(u/w_{M+1-j})L_{a k}(u/w_{M+1-k})
\widetilde{R}_{jk}(w_{M+1-j}/w_{M+1-k})
\label{RLL2}
\end{align}
holding in 
$\End(W_{a} \otimes V_j \otimes V_k$). Here $\widetilde{R}$
is given by
\begin{align}
\widetilde{R}(v)
=
\begin{pmatrix}
1 & 0 & 0 & 0 \\
0 & \beta(v-1) & v & 0 \\
0 & 1 & 0 & 0 \\
0 & 0 & 0 & v
\end{pmatrix},
 \label{Rtildematrix}
\end{align}
which intertwines the $L$-operators acting on a common auxiliary space
(but acting on  different quantum spaces).
Note the usual $RLL$-relation \eqref{RLL2} intertwines
the $L$-operators acting on a same quantum space but acting on different auxiliary spaces.
The above $RLL$-relation \eqref{RLL2} allows one to construct
the monodromy matrix as a product of the $L$-operators
acting on the same quantum space (see also  Appendix
for an example of using its property
to examine the symmetries of the domain wall boundary partition function),
and rewriting the intermediate scalar products in terms of the
resultant monodromy matrices makes one see Property 1 holds.

Property 2 can be shown by inserting the completeness relation
into the intermediate scalar products 
\begin{align}
S(\{ u \}_n | \{ v \}_N| \{ w \})&=\langle 
1^{N-n} 0^{M-N+n}|
\prod_{j=1}^n \mathcal{C}(u_j,\{w\}) \prod_{k=1}^N \mathcal{B}(v_k,\{w\})|\Omega \rangle
\nonumber \\
&=\sum_{k=1}^{M-N+n}
\langle 
1^{N-n} 0^{M-N+n}|\mathcal{C}(u_n,\{w\})
|1^{N-n} 0^{M-N+n-k} 1 0^{k-1} 
\rangle
\nonumber \\
&\qquad\times
\langle 1^{N-n} 0^{M-N+n-k} 1 0^{k-1}|
\prod_{j=1}^{n-1}\mathcal{C}(u_{j},\{w\}) \prod_{k=1}^N \mathcal{B}(v_k,\{w\})|\Omega \rangle,
\label{genericrecursiverelation}
\end{align}
and noting that the factor containing $u_n$ is calculated as
\begin{align}
&\langle 
1^{N-n }0^{M-N+n}|
\mathcal{C}(u_n,\{w\})
| 1^{N-n} 0^{M-N+n-k} 1 0^{k-1} \rangle \nn \\
&\qquad \qquad =\prod_{j=M-N+n+1}^M \Bigg(1+\frac{\beta u_n}{w_j} \Bigg)
\prod_{j=k+1}^{M-N+n} \Bigg(1-\frac{\beta u_n}{w_j} \Bigg)
\prod_{j=1}^{k-1} \frac{u_n}{w_j}.
\end{align}

Property 3 can be obtained by setting $u_n=\beta^{-1} w_{M-N+n}$
in \eqref{genericrecursiverelation},
or can be directly observed by its graphical representation
that the top row is completely frozen.

To show Property 4, we first note by
the graphical representation (see figure~\ref{initialpic}) that
\begin{align}
S(\{ u \}_0|\{ v \}_N|\{ w \})
=
\prod_{j=1}^N \prod_{k=1}^{M-N} \Bigg(1-\frac{\beta v_j}{w_k} \Bigg)
\langle 1^N|\prod_{j=1}^N \mathcal{B}_N(v_j)|0^N \rangle. 
\label{simplification}
\end{align}
where 
\begin{align}
\mathcal{B}_N(v_k)={}_a \langle 0| \prod_{j=1}^N 
\mathcal{L}_{aj}(v_k/w_{M+1-j}) |1 \rangle_a.
\label{B-domain}
\end{align}
One can evaluate the domain wall boundary partition function
$\langle 1^N|\prod_{j=1}^N \mathcal{B}_N(v_j)|0^N \rangle$
by the standard procedure following the arguments of Izergin-Korepin
\cite{Kore,Ize},
which is given in Appendix.
\begin{figure}[ttt]
\begin{center}
\includegraphics[width=0.99\textwidth]{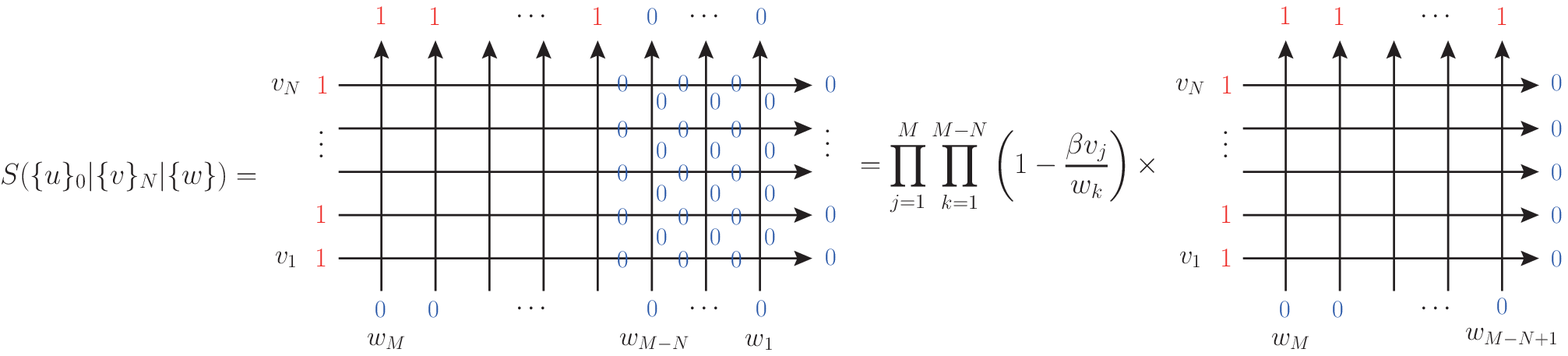}
\end{center}
\caption{The intermediate scalar products \eqref{initial} for $n=0$,
which corresponds to the domain wall boundary partition function.}
\label{initialpic}
\end{figure}
The result has the following simple factorized form
\begin{align}
\langle 1^N|\prod_{j=1}^N \mathcal{B}_N(v_j)|0^N \rangle
=\prod_{j=1}^N \frac{2v_j}{w_{M-N+j}^{j}}
\prod_{1 \le j<k \le N}(v_j+v_k), \label{domainwallboundary}
\end{align}
which together with \eqref{simplification} proves Property 4.
\end{proof}

\begin{lemma}\label{uniqueness}
The properties in Lemma~\ref{property} uniquely determine
the intermediate scalar product \eqref{intermediatedef}.
\end{lemma}

Due to Lemma~\ref{uniqueness}, the following determinant
representation for the intermediate scalar product is valid.
\begin{theorem}{\label{determinantthm}}
The intermediate scalar product $S(\{ u \}_n|\{ v \}_N|\{ w \})$ \eqref{intermediatedef} 
has the following determinant form:
\begin{align}
S(\{ u \}_n|\{ v \}_N|\{ w \}) 
=&
\prod_{j=1}^N (2v_j)
\prod_{1 \le j<k \le N} \frac{v_j+v_k}{v_k-v_j}
\prod_{M-N+n+1 \le j<k \le M} \frac{1}{\beta(1-w_k/w_j)}
\nn \\
&\times
\prod_{1 \le j<k \le n} \frac{u_j+u_k}{u_j-u_k}
\mathrm{det}_N Q(\{ u \}_n|\{ v \}_N|\{ w \})
\label{intermediatedeterminant}
\end{align}
with an $N \times N$ matrix
$Q(\{ u \}_n|\{ v \}_N|\{ w \})$ whose matrix elements
are given by
\begin{align}
&Q(\{ u \}_n|\{ v \}_N|\{ w \})_{jk} \nn \\
& =
\begin{cases}
\displaystyle \prod_{l=M-N+n+1}^M \frac{\displaystyle \beta u_j+w_l}
{\displaystyle \beta u_j-w_l}
\frac{\displaystyle
a(u_j,\{w\})d(v_k,\{w\})
-a(v_k,\{w\})d(u_j,\{w\})
}
{\displaystyle
v_k-u_j
},  & \text{ ($1\le j \le n$)}  \\
\displaystyle
\frac{1}{w_{M-N+j}} \prod_{\substack{l=1 \\ l \neq M-N+j}}^M
\left( 1-\frac{\beta v_k}{w_l} \right),
& \text{ ($n+1\le j \le N$)}
\end{cases}.
\end{align}
\end{theorem}

\begin{proof}
We can directly see that the determinant formula \eqref{intermediatedeterminant}
satisfies all the properties in Lemma~\ref{property}. 
To show Property 4, we utilize  the Cauchy determinant formula
\begin{align}
\mathrm{det}_N \left( \frac{1}{x_j-y_k} \right)
=\frac{\prod_{1 \le j<k \le N}(x_k-x_j)(y_j-y_k)}{\prod_{j,k=1}^N (x_j-y_k)}.
\end{align}

Finally due to Lemma~\ref{uniqueness},   the determinant formula 
\eqref{intermediatedeterminant} holds.
\end{proof}

\begin{corollary}
Taking $n=N$ in \eqref{intermediatedeterminant} yields the determinant
representation of the scalar product for the six-vertex model
with inhomogeneous
parameters \eqref{SP}:
\begin{align}
\bra \Psi(\{ u \}_N,\{w\})| & \Psi(\{ v \}_N,\{w\}) \ket \nn \\
&=
\prod_{j=1}^N (2v_j)
\prod_{1 \le j<k \le N }
\frac{(u_j+u_k)(v_j+v_k)}{(u_j-u_k)(v_k-v_j)}
\mathrm{det}_N Q(\{ u \}_N|\{ v \}_N|\{ w \})
\end{align}
with
\begin{align}
Q(\{ u \}_N|\{ v \}_N|\{ w \})_{jk} 
=
\frac{
a(u_j,\{w\})d(v_k,\{w\})
-a(v_k,\{w\})d(u_j,\{w\})
}
{
v_k-u_j
}.
\end{align}
Further taking the homogeneous limit $w_j\to 1$ ($1\le j \le n$) 
yields \eqref{generalscalar} in 
Theorem~\ref{determinantthm}. 
\end{corollary}

\section{Matrix elements of the $t$-deformed boson model and the six-vertex model}
In this section, we derive matrix elements for the $t$-deformed boson model
for the generic parameter $t$, and then we restrict ourselves to the case of the
six-vertex model $t=-1$.

Consider the arbitrary off-shell
state, i.e., the parameters $\{v\}_N$ in the $N$-particle state
\eqref{statevector-B} are arbitrary.
The orthonormal basis of
the $N$-particle state $|\Psi(\{v \}_N) \ket$
and its dual $\bra \Psi(\{v \}_N)|$
is given by $| \{ n \}_{M,N} \ket := |n_1\ket_1\
\otimes \dots \otimes |n_{M}\ket_{M}$
and $\bra \{ n \}_{M,N}| := _{1}\bra n_1| \otimes \dots
\otimes _{M} \bra n_{M}|$, where $n_1+\cdots+n_{M}=N$.
The wavefunctions can be expanded in this basis as
\begin{align}
&|\Psi(\{v \}_N) \ket
=\sum_{\substack{0 \le n_1,\dots,n_{M} \le N \\
n_1+\cdots+n_{M}=N}}
\bra \{ n \}_{M,N}|\psi(\{v \}_N) \ket
|\{ n \}_{M,N} \ket, \\
&\bra \Psi(\{v \}_N)|
=\sum_{\substack{0 \le n_1, \dots,n_{M} \le N \\
n_1+\cdots+n_{M}=N}}
\bra \{ n \}_{M,N}| \bra \psi(\{v \}_N)| \{ n \}_{M,N} \ket.
\end{align}
There is a one-to-one correspondence between
the set $\{ n \}_{M,N}=\{n_1,\dots,n_{M} \}$ ($n_1+\cdots+n_{M}=N$)
and the Young diagram $x=(x_1,x_2,\dots,x_N)$
($M \ge x_1 \ge x_2 \ge \cdots \ge x_N \ge 1$).
Namely, each Young diagram $x$ under the constraint
$x_1 \le M$, $\ell(x)=N$
can be labeled by a set of integers $\{ n \}_{M,N}$ as
$x=(M^{n_{M}}, \dots, 1^{n_1})$.

The following definition \cite{Wh}
on the ordering on the basis of particle configurations
is useful for later purpose.
\begin{definition} \cite{Wh}
For two configurations
$\{ m \}_{M,N+1}=\{m_1,\dots,m_{M} \}$ $(m_1+\cdots+m_{M}=N+1)$
and
$\{ n \}_{M,N}=\{n_1,\dots,n_{M} \}$ $(n_1+\cdots+n_{M}=N)$,
let $\sum_j^m=\sum_{k=j}^{M}m_k$ and $\sum_j^n=\sum_{k=j}^{M}n_k$.
We say that the particle configurations
$\{ m \}_{M,N+1}$  and $\{ n \}_{M,N}$ are admissible, if and only if
$0 \le (\sum_j^m-\sum_j^n) \le 1$ ($1 \le j \le M$), and write this relation
as $\{ m \}_{M,N+1} \triangleright \{ n \}_{M,N}$.
\end{definition}
Moreover we also define the ordering on the Young diagrams.
\begin{definition}
For two Young diagrams $y=(y_1,y_2,\dots,y_{N+1})$ and
$x=(x_1,x_2,\dots,x_N)$,
we say that $y$ and $x$ interlace,
if and only if $y_j \ge x_j \ge y_{j+1}$ $(j=1,\dots,N)$,
and write this relation as $y \succ x$. 
\end{definition}
\begin{proposition}\label{admissible}
Let $\{m\}_{M,N+1}$ and $\{n\}_{M,N}$ be the particle configurations
described by the Young diagram $y=(y_1,\dots,y_{N+1})$
and $x=(x_1,\dots,x_{N})$. Then 
\begin{align}
y \succ x \Longleftrightarrow \{ m \}_{M,N+1} \triangleright \{ n \}_{M,N}.
\end{align}
\end{proposition}
\noindent
For $\{ m \}_{M,N+1}$ and $\{ n \}_{M,N}$,
we introduce $\{ p \}_r=\{1 \le p_1<\cdots<p_r \le M \}$ to be the set of all integers
$p$ such that $m_p=n_p+1$, and
$\{ q \}_s=\{1 \le q_1<\cdots<q_s \le M \}$ to be the set of all integers
$q$ such that $m_q+1=n_q$.
When $\{ m \}_{M,N+1}$ and $\{ n \}_{M,N}$ satisfy the
admissible condition $\{ m \}_{M,N+1} \triangleright \{ n \}_{M,N}$,
$\{ p \}_r$ and $\{ q \}_s$ satisfy
$s=r-1$ and $p_k < q_k < p_{k+1}$ ($k=1,\dots,r-1$).

From the matrix elements of the $L$-operator, one finds
\begin{align}
\bra \{ m \}_{M,N+1}|\mathcal{B}(v)| \{ n \}_{M,N} \ket=0, \ \ \ 
\mathrm{unless} \ \ \ \{ m \}_{M,N+1} \triangleright \{ n \}_{M,N}.
\end{align}
When the admissible condition is satisfied, one finds the following.

\begin{align}
\langle \{m \}_{M,N+1}|
\mathcal{B}(v)|\{ n \}_{M,N} \rangle
=&{}_{a} \bra 0 | \bra \{m\}_{M,N+1}| \prod_{j=1}^{M} \mathcal{L}_{aj}(v)
|1 \ket_{a} |\{ n \}_{M,N} \ket \nonumber \\
=&
v^{\sum_{j=1}^r p_j-\sum_{j=1}^{r-1} q_j}
\prod_{j=1}^r (1-t^{n_{p_j}+1}) \prod_{j=1}^r \prod_{k=p_j+1}^{q_j-1}
(1-\beta v t^{n_k}), \label{tgeneric}
\end{align}
where $q_0=0, \ q_r=M+1$. \\
This can be shown by combining the following partial actions:
\begin{align}
&\prod_{l=q_{j-1}+1}^{q_j} \mathcal{L}_{al}(v)| 1 \ket_a \otimes
\left\{ \otimes_{k=q_{j-1}+1}^{q_j}
| n_k \ket_k \right\} \nonumber \\
&\qquad =v^{p_j-q_{j-1}} (1-t^{n_{p_j}+1})
\prod_{l=p_j+1}^{q_j-1}(1-\beta v t^{n_l})
| 1 \ket_a \left\{ \otimes_{k=q_{j-1}+1}^{q_j}
| m_k \ket_k \right\} \quad (1 \le j \le r-1), \nonumber \\
&\prod_{l=q_{r-1}+1}^{M} \mathcal{L}_{al}(v)| 1 \ket_a  \otimes
\left\{ \otimes_{k=q_{r-1}+1}^{M}
| n_k \ket_k \right\} \nonumber \\
&\qquad =v^{p_r-q_{r-1}} (1-t^{n_{p_r}+1})
\prod_{l=p_r+1}^{M}(1-\beta v t^{n_l})
| 0 \ket_a \left\{ \otimes_{k=q_{r-1}+1}^{M}
| m_k \ket_k \right\}.
\end{align}

Next, we examine the matrix elements of the one-row $B$ and $C$ operators
furthermore at the point $t=-1$.
We first reduce the matrix elements \eqref{tgeneric} to a simpler form and then
translate into the language of Young diagrams.
The result for the
matrix elements in the language of Young diagrams
can be summarized as follows.
\begin{proposition}
The matrix elements 
$\langle \{ m \}_{M,N+1}|\mathcal{B}(v)|\{ n \}_{M,N} \rangle$ at $t=-1$
are given by
\begin{align}
&\langle \{ m \}_{M,N+1}|\mathcal{B}(v)|\{ n \}_{M,N} \rangle
\nonumber \\
&\qquad =\frac{1+2 \beta z}{(1+\beta z)^{M+1}}
\left\{ \frac{2(1+\beta z)}{1+2 \beta z} \right\}^{\#(y|x)}
z^{\sum_{j=1}^{N+1} y_j-\sum_{j=1}^N x_j}
\prod_{j=1}^N \{1+2 \beta z(1-\delta_{x_j y_{j+1}}) \}.
\label{tminusoneskew}
\end{align}
when $x=(x_1,\dots,x_N)=(M^{n_M},\dots,1^{n_1})
\subset M^N$
and $y=(y_1,\dots,y_{N+1})=(M^{m_M},\dots,1^{m_1})
\subset M^{N+1}$
are strict partitions satisfying $y \succ x$,
and zero otherwise. Here $\#(y|x)$ denotes the 
number of parts in $y$ which are not in $x$.
We regard the product in the right hand side of \eqref{tminusoneskew}
as 1 when $x=\phi$.
\end{proposition}

\begin{proposition}
The matrix elements 
$\langle \{ n \}_{M,N}|\mathcal{C}(v)|\{ m \}_{M,N+1} \rangle$ at $t=-1$
are given by
\begin{align}
&\langle \{ n \}_{M,N}|\mathcal{C}(v)|\{ m \}_{M,N+1} \rangle
\nonumber \\
=&\frac{1}{z(1+\beta z)^{M-1}}
\left\{ \frac{2(1+\beta z)}{1+2 \beta z} \right\}^{\#(y^\vee|x^\vee)-1}
z^{\sum_{j=1}^{N+1} y_j^\vee-\sum_{j=1}^N x_j^\vee}
\prod_{j=1}^N \{1+2 \beta z(1-\delta_{x_j^\vee y_{j+1}^\vee}) \}.
\label{tminusoneskewc}
\end{align}
when $x^\vee=(M+1)^N / x=(x_1^\vee,\dots,x_N^\vee)
=(M^{n_1},\dots,1^{n_M})
\subset M^N$
and $y^\vee=(M+1)^{N+1} / y=(y_1^\vee,\dots,y_{N+1}^\vee)
=(M^{m_1},\dots,1^{m_M})
\subset M^{N+1}$
are strict partitions
satisfying $y^\vee \succ x^\vee$,
and zero otherwise.
We regard the product in the right hand side of \eqref{tminusoneskewc}
as 1 when $x^\vee=\phi$.
\end{proposition}
\begin{proof}
We show \eqref{tminusoneskew}. Eq.~\eqref{tminusoneskewc} can be calculated in the
same way. 

First, we note that
since $\prod_{j=1}^r(1-t^{n_{p_j}+1}) \neq 0$ unless 
$n_{p_j}=1 \ (\mathrm{mod} \ 2)$ for all $j$,
we only need to consider the case when
both $\{ m \}_{M,N+1}$ and $\{ n \}_{M,N}$ are sequences of numbers 0 and 1.
Otherwise, the matrix elements vanish as explained in section~2.
Translating this restriction to the language of Young diagrams,
this means that we restrict both the partitions $y$ and $x$
to be strict.

The matrix elements for the case $y \succ x$
is calculated as
\eqref{tgeneric} for generic $t$,
which can be furthermore simplified at $t=-1$ as follows
\begin{align}
\langle \{ m \}_{M,N+1}|
\mathcal{B}(v)|\{ n \}_{M,N} \rangle
=&2^r
v^{\sum_{j=1}^r p_j-\sum_{j=1}^{r-1} q_j}
(1-\beta v)^{\# \{k \in \cup_{j=1}^{r} \{ p_j+1, \dots , q_j-1 \} |n_k=0 \}} \nonumber \\
&\times
(1+\beta v)^{\# \{k \in \cup_{j=1}^{r} \{ p_j+1, \dots , q_j-1 \} |n_k=1 \}},
\end{align}
where $\{ m \}_{M,N+1}$ and $\{ n \}_{M,N}$ are both sequences of 0 and 1
satisfying $\{ m \}_{M,N+1} \triangleright \{ n \}_{M,N}$. \\
Using
\begin{align}
&(1-\beta v)^{\# \{k \in \cup_{j=1}^{r} \{ p_j+1, \dots , q_j-1 \} 
|n_k=0 \}} \nonumber \\
&\qquad \qquad =(1-\beta v)^{\# \{k \in \cup_{j=1}^{r} \{ p_j+1, \dots , q_j-1 \} 
\}}
(1-\beta v)^{-\# \{k \in \cup_{j=1}^{r} \{ p_j+1, \dots , q_j-1 \} 
|n_k=1 \}} \nonumber \\
&\qquad \qquad=(1-\beta v)^{
\sum_{j=1}^{r-1} q_j-\sum_{j=1}^r p_j+M+1-r
-\# \{k \in \cup_{j=1}^{r} \{ p_j+1, \dots , q_j-1 \} 
|n_k=1 \}},
\end{align}
one has
\begin{align}
\langle \{ m \}_{M,N+1}|\mathcal{B}(v)|&\{ n \}_{M,N} \rangle \nonumber \\
=&2^r
v^{\sum_{j=1}^r p_j-\sum_{j=1}^{r-1} q_j}
(1-\beta v)^{
\sum_{j=1}^{r-1} q_j-\sum_{j=1}^r p_j+M+1-r}
\nonumber \\
&\times \left(\frac{1+\beta v}{1-\beta v} 
\right)^{\# \{k \in \cup_{j=1}^{r} \{ p_j+1, \dots , q_j-1 \} |n_k=1 \}} \nonumber \\
=&2^r
(1-\beta v)^{M+1-r}
(v^{-1}-\beta)^{\sum_{j=1}^{r-1} q_j-\sum_{j=1}^r p_j}
\nonumber \\
&\times \left(\frac{1+\beta v}{1-\beta v} 
\right)^{\# \{k \in \cup_{j=1}^{r} \{ p_j+1, \dots , q_j-1 \} |n_k=1 \}}.
\end{align}
From the translation rule \cite{MS}
\begin{align}
&\sum_{j=1}^r p_j-\sum_{j=1}^{r-1}q_j
=\sum_{j=1}^{N+1} y_j-\sum_{j=1}^N x_j, \nn \\
&r-1+\# \{k \in \cup_{j=1}^{r} \{ p_j+1, \dots , q_j-1 \} |n_k \neq 0 \}
=\# \{ j \in \{1, \dots , N \}|x_j \neq y_{j+1} \},
\end{align}
One gets
\begin{align}
\langle \{ m \}_{M,N+1}|& \mathcal{B}(v)|\{ n \}_{M,N} \rangle \nonumber \\
=&2^r
(1-\beta v)^{M-r+1}
(v^{-1}-\beta)^{\sum_{j=1}^N x_j-\sum_{j=1}^{N+1} y_j}
\left(\frac{1+\beta v}{1-\beta v} 
\right)^{\# \{ j \in \{1, \dots , N \}|x_j \neq \mu_{j+1} \}-r+1} \nonumber \\
=&2^r
\left( \frac{1+\beta v}{1-\beta v} \right)^{1-r}
(1-\beta v)^{M+1-r}
(v^{-1}-\beta)^{\sum_{j=1}^N x_j-\sum_{j=1}^{N+1} y_j}
\nonumber \\
&\times \prod_{j=1}^N
\left\{1+\frac{2 \beta v}{1-\beta v}(1-\delta_{x_j,y_{j+1}})
\right\}.
\end{align}
Introducing the variable $z=(v^{-1}-\beta)^{-1}$ and changing the variable
from $v$ to $z$,
the matrix elements can be rewritten as
\begin{align}
&\langle \{ m \}_{M,N+1}|\mathcal{B}(v)|\{ n \}_{M,N} \rangle \nonumber \\
&\qquad =2^r
(1+2 \beta z)^{1-r} (1+\beta z)^{r-M-1}
z^{\sum_{j=1}^{N+1} y_j-\sum_{j=1}^N x_j}
\prod_{j=1}^N \{1+2 \beta z(1-\delta_{x_j,y_{j+1}})
\},
\end{align}
where
$r=:\#(y|x)$ is the number of parts in $y$
which are not in $x$.
\end{proof}
Note that 
the elements of the $L$-operator $\mathcal{L}_{aj}(v)$ \eqref{Lop1}--\eqref{Lop2} 
at $t=-1$ reduces to those for the six-vertex model \eqref{sixvertexLoperator}:
\begin{align}
\mathcal{L}_{aj}(v)
=
\begin{pmatrix}
1-\beta v & 0 & 0 & 0 \\
0 & 1+\beta v & 2v & 0 \\
0 & 1 & v & 0 \\
0 & 0 & 0 & v
\end{pmatrix}_{aj}=
\begin{pmatrix}
\frac{1}{1+\beta z} & 0 & 0 & 0 \\
0 & \frac{1+2\beta z}{1+\beta z}  &\frac{2 z}{1+\beta z} & 0 \\
0 & 1 & \frac{z}{1+\beta z} & 0 \\
0 & 0 & 0 &\frac{z}{1+\beta z}
\end{pmatrix}_{aj}. 
\label{sixvertexLoperator2}
\end{align}

\begin{example}
We set $M=5$ and $N=2$
and consider the case
$\{ m \}_{M,N+1}=\{10011 \}$ and $\{ n \}_{M,N}=\{00110 \}$.
Translating into the language of Young diagrams, we have
$y=(5,4,1)$ and $x=(4,3)$.
One finds $\#(y|x)=2$
since 5 and 1 are in $y$ but not in $x$.
We also have
$\sum_{j=1}^3 y_j-\sum_{j=1}^2 x_j=3$
and $x_1=y_2$, $x_2 \neq y_3$.
The right hand side of \eqref{tminusoneskew}
becomes $2^2 z^3 (1+\beta z)^{-4}$ which can be easily checked to
match with the left hand side.
\end{example}
For $\beta=0$, the matrix elements reduce to the
skew Schur $Q$-polynomials \cite{Ts,Kor}
\begin{align}
&\langle \{ m \}_{M,N+1}|\mathcal{B}(v)|\{ n \}_{M,N} \rangle
=2^{\#(y|x)}
z^{\sum_{j=1}^{N+1} y_j-\sum_{j=1}^N x_j}=Q_{y/x}(z), \label{skewspecial}
\end{align}
which we will use to derive the wavefunction in the next section.

\section{Wavefunctions of the six-vertex model}
Let us examine the wavefunction at the point $t=-1$ where the
$t$-boson model reduces to the six-vertex model.
We show the corresponding wavefunction is essentially the Schur polynomials.
\begin{definition}
The Schur polynomial is defined to be the following determinant:
\begin{align}
s_\lambda(\bs{z})=
   \frac{\mathrm{det}_N(z_j^{\lambda_k+N-k})}
        {\prod_{1 \le j < k \le N}(z_j-z_k)},
 \label{Schur}
\end{align}
where $z=\{z_1,\dots,z_N \}$ is a set of variables
and $\lambda$ denotes a Young diagram
$\lambda=(\lambda_1,\lambda_2,\dots,\lambda_N)$
with weakly decreasing non-negative integers
$\lambda_1 \ge \lambda_2 \ge \cdots \ge \lambda_N \ge 0$.
\end{definition}
We show the following equivalence
between the wavefunction of the six-vertex model
and the Schur polynomials.
\begin{theorem}
The wavefunction has the following form
\begin{align}
\langle \{ m \}_{M,N}|\prod_{j=1}^N \mathcal{B}(v_j)|\Omega \rangle
=\frac{2^N \prod_{j=1}^N z_j \prod_{1 \le j < k \le N}(z_j+z_k+2 \beta z_j z_k)}{\prod_{j=1}^N (1+\beta z_j)^M} s_\lambda(\bs{z}). \label{wavefunctiontheorem}
\end{align}
where $z_j=(v_j^{-1}-\beta)^{-1}$ and
$\lambda=(\lambda_1,\dots,\lambda_N) \subset (M-N)^N$ is a partition
related to $x=(x_1,\dots,x_N)=(M^{m_M},\dots,1^{m_1})$
by $\lambda_j=x_{j}-N+j-1$.

The dual wavefunction has the following form
\begin{align}
\langle \Omega |\prod_{j=1}^N \mathcal{C}(v_j)|\{ m \}_{M,N} \rangle
=\frac{\prod_{1 \le j < k \le N}(z_j+z_k+2 \beta z_j z_k)}
{\prod_{j=1}^N (1+\beta z_j)^{M-1}} s_{\lambda^\vee}(\bs{z}),
\label{dualwavefunctiontheorem}
\end{align}
where $\lambda^\vee$ is the Poincare dual of $\lambda$:
$\lambda^\vee=(\lambda_1^\vee,\dots,\lambda_N^\vee), \ 
\lambda_j^\vee=M-N-\lambda_{N+1-j}$.
\end{theorem}
We remark that the integrable model we consider here
seems to be a special case of the one considered in \cite{Ta},
whose wavefunction was obtained
by the coordinate Bethe ansatz up to a normalization factor.
However, to prove combinatorial formulae,
it is crucial to determine the exact form starting from the first principle, i.e.,
starting from the $L$-operator,
since we combine the above theorem with
the exact expression for the matrix elements
\eqref{tminusoneskew} and \eqref{tminusoneskewc}
derived in the previous section to derive a new combinatorial
formula for example.
\begin{proof}
We show \eqref{wavefunctiontheorem}.
Eq.~\eqref{dualwavefunctiontheorem} can be proved in the same way.
First, we redefine the $L$-operator as
\begin{align}
\widetilde{\mathcal{L}}_{aj}(z)
=(1+\beta z)\mathcal{L}_{aj}(v)
=
\begin{pmatrix}
1 & 0 & 0 & 0 \\
0 & 1+2 \beta z & 2z & 0 \\
0 & 1+\beta z & z & 0 \\
0 & 0 & 0 & z
\end{pmatrix}_{aj}, 
\label{redifinedsixvertexLoperator}
\end{align}
and the corresponding monodromy matrix
\begin{align}
\widetilde{\mathcal{T}}_{a}(z)=\prod_{j=1}^{M} \widetilde{\mathcal{L}}_{a j}(z)
=
\begin{pmatrix}
\widetilde{\mathcal{A}}(z) & \widetilde{\mathcal{B}}(z)  \\
\widetilde{\mathcal{C}}(z) & \widetilde{\mathcal{D}}(z)
\end{pmatrix}_{a}, 
\end{align}
and show the following equivalent equality for \eqref{wavefunctiontheorem}
\begin{align}
\langle \{ m \}_{M,N}|\prod_{j=1}^N \widetilde{\mathcal{B}}(z_j)|\Omega \rangle
=2^N \prod_{j=1}^N z_j \prod_{1 \le j < k \le N}(z_j+z_k+2 \beta z_j z_k) s_\lambda(\bs{z}).
\label{equivalentexpressionwavefunction}
\end{align}
To prove this, we first show the following lemma
\begin{lemma}
\begin{align}
\frac{\langle \{ m \}_{M,N}|\prod_{j=1}^N \widetilde{\mathcal{B}}(z_j)|\Omega \rangle}{\prod_{1 \le j < k \le N}(z_j+z_k+2 \beta z_j z_k)},
\end{align}
does not depend on $\beta$.
\label{independence}
\end{lemma}
\begin{proof}
We prove this lemma by showing the following properties
for  $
\langle \{ m \}_{M,N}|\prod_{j=1}^N \widetilde{\mathcal{B}}(z_j)|\Omega \rangle
$: 
\begin{enumerate}
\item 
$
\langle \{ m \}_{M,N}|\prod_{j=1}^N \widetilde{\mathcal{B}}(z_j)|\Omega \rangle
$
is a polynomial of $\beta$ with highest degree $N(N-1)/2$. 
\item
$
\langle \{ m \}_{M,N}|\prod_{j=1}^N \widetilde{\mathcal{B}}(z_j)|\Omega \rangle
$
has $z_j+z_k+2 \beta z_j z_k$, $1 \le j < k \le N$ as factors. 
\end{enumerate}
We first show
$
\mathrm{deg}_\beta
\langle \{ m \}_{M,N}|\prod_{j=1}^N \widetilde{\mathcal{B}}(z_j)|\Omega \rangle
\le N(N-1)/2
$ by induction on $N$.
The case $N=1$ follows as an special case of the general fact
$
\mathrm{deg}_\beta
\langle \{ m \}_{M,{N+1}}|\widetilde{\mathcal{B}}(z)|\{ n \}_{M,N} \rangle
\le N
$
which can be seen easily from the definition of
the $L$-operator $\widetilde{\mathcal{L}}(z)$.
Next, let us assume $
\mathrm{deg}_\beta
\langle \{ m \}_{M,N}|\prod_{j=1}^N \widetilde{\mathcal{B}}(z_j)|\Omega \rangle
\le N(N-1)/2
$.
One can see
$
\mathrm{deg}_\beta
\langle \{ m \}_{M,N+1}|\prod_{j=1}^{N+1} \widetilde{\mathcal{B}}(z_j)|\Omega \rangle
\le (N+1)N/2
$
by  combining the assumption
$
\mathrm{deg}_\beta
\langle \{ m \}_{M,N}|\prod_{j=1}^N \widetilde{\mathcal{B}}(z_j)|\Omega \rangle
\le N(N-1)/2
$,
the fact
$
\mathrm{deg}_\beta
\langle \{ m \}_{M,{N+1}}|\widetilde{\mathcal{B}}(z)|\{ n \}_{M,N} \rangle
\le N
$
and the decomposition
\begin{align}
\langle \{ m \}_{M,N+1}|\prod_{j=1}^{N+1} \widetilde{\mathcal{B}}(z_j)|\Omega \rangle
=
\sum_{\{ n \}_{M,N}}
\langle \{ m \}_{M,{N+1}}|\widetilde{\mathcal{B}}(z_{N+1})|\{ n \}_{M,N}
\rangle
\langle \{ n \}_{M,N}|\prod_{j=1}^{N} \widetilde{\mathcal{B}}(z_j)|\Omega
\rangle.
\end{align}
Next, we show Property 2.
It is enough to show the case $N=2$.
The case for generic $N$ follows from
the commutativity ${[} \widetilde{\mathcal{B}}(z_j), \widetilde{\mathcal{B}}(z_k) {]}=0$ and the decomposition
\begin{align}
\langle \{ m \}_{M,N}|\prod_{j=1}^{N} \widetilde{\mathcal{B}}(z_j)|\Omega \rangle
=
\sum_{\{ n \}_{M,2}}
\langle \{ m \}_{M,N}|\prod_{j=3}^N \widetilde{\mathcal{B}}(z_j)|\{ n \}_{M,2}
\rangle
\langle \{ n \}_{M,2}| \widetilde{\mathcal{B}}(z_1) \widetilde{\mathcal{B}}(z_2)|\Omega
\rangle.
\end{align}
Now we show for the case $N=2$ by induction on $M$.
Let us denote the $A, B ,C, D$ operators
consisting of $M$ $L$-operators as $\widetilde{\mathcal{B}}_M(z)$ for example.
The case $M=2$ can be checked explicitly
$\langle 1,1|\widetilde{\mathcal{B}}_M(z_1) \widetilde{\mathcal{B}}_M(z_2)| \Omega \rangle =4z_1 z_2(z_1+z_2+2 \beta z_1 z_2)$
.
Let us assume that
$\langle x_1, x_2|\widetilde{\mathcal{B}}_M(z_1) \widetilde{\mathcal{B}}_M(z_2)| \Omega \rangle $
has $z_1+z_2+2 \beta z_1 z_2$ as a factor.
We examine
$\langle x_1, x_2|\widetilde{\mathcal{B}}_{M+1}(z_1) 
\widetilde{\mathcal{B}}_{M+1}(z_2)| \Omega \rangle $.
One can easily show by its graphical description that
\begin{align}
\langle x_1, x_2|\widetilde{\mathcal{B}}_{M+1}(z_1) 
\widetilde{\mathcal{B}}_{M+1}(z_2)| \Omega \rangle
=(z_1 z_2)^{x_1-1}
\langle x_1, x_2|\widetilde{\mathcal{B}}_{x_2-x_1+1}(z_1)
\widetilde{\mathcal{B}}_{x_2-x_1+1}(z_2)| \Omega \rangle.
\end{align}
If $x_1 \neq 1$ or $x_2 \neq M+1$,
$\langle x_1, x_2|\widetilde{\mathcal{B}}_{M+1}(z_1) 
\widetilde{\mathcal{B}}_{M+1}(z_2)| \Omega \rangle$ has
$z_1+z_2+2 \beta z_1 z_2$ as a factor by assumption.

We examine the remaining case $x_1=1, x_{M+1}=M+1$.
We show
\begin{align}
\langle 1, M+1|\widetilde{\mathcal{B}}_{M+1}(z_1) 
\widetilde{\mathcal{B}}_{M+1}(z_2)| \Omega \rangle
=4 z_1 z_2 \frac{z_1^M-z_2^M}{z_1-z_2} (z_1+z_2+2 \beta z_1 z_2).
\label{specialcase}
\end{align}
By graphical description, one sees
\begin{align}
\langle x_1, x_2|\widetilde{\mathcal{B}}_{M+1}(z_1) 
\widetilde{\mathcal{B}}_{M+1}(z_2)| \Omega \rangle
=2z_1 \{f_M(z_1,z_2)+2 z_2^{M+1}(1+2 \beta z_1)  \},
\label{intermediaterelation}
\end{align}
where $f_M(z_1,z_2)=\langle 1|\widetilde{\mathcal{D}}_M(z_1)
\widetilde{\mathcal{B}}_M(z_2) | \Omega \rangle$.
Again, we use its graphical representation to derive the
following recursive relation
\begin{align}
f_M(z_1,z_2)=z_1 \{f_{M-1}(z_1,z_2)+4 z_2^M(1+ \beta z_1) \},
\label{recursivef}
\end{align}
with the initial condition
\begin{align}
f_2(z_1,z_2)=2z_1 z_2(z_1+2 z_2+2 \beta z_1 z_2).
\label{initialf}
\end{align}
We can show by induction that
\begin{align}
f_M(z_1,z_2)=2 z_2 \frac{z_1^M-z_2^M}{z_1-z_2}(z_1+z_2+2 \beta z_1 z_2)
-2 z_2^{M+1}(1+2 \beta z_1), \label{use}
\end{align}
solves the recursive relation \eqref{recursivef} and
the initial condition \eqref{initialf}.
Hence, the expression \eqref{specialcase} follows from
\eqref{intermediaterelation} and \eqref{use}.
We thus have shown by induction that \\
$\langle x_1, x_2|\widetilde{\mathcal{B}}_{M+1}(z_1) 
\widetilde{\mathcal{B}}_{M+1}(z_2)| \Omega \rangle$
has
$z_1+z_2+2 \beta z_1 z_2$ as a factor, and Property 2 is proved.

From Property 2 we have,
$
\mathrm{deg}_\beta
\langle \{ m \}_{M,N}|\prod_{j=1}^N \widetilde{\mathcal{B}}(z_j)|\Omega \rangle
\ge N(N-1)/2
$.
Together with \\
$
\mathrm{deg}_\beta
\langle \{ m \}_{M,N}|\prod_{j=1}^N \widetilde{\mathcal{B}}(z_j)|\Omega \rangle
\le N(N-1)/2
$ which we proved before, we have Property 1.
\end{proof}
From Lemma \ref{independence}, one can examine
\begin{align}
\frac{\langle \{ m \}_{M,N}|\prod_{j=1}^N \widetilde{\mathcal{B}}(z_j)|\Omega \rangle}{\prod_{1 \le j < k \le N}(z_j+z_k+2 \beta z_j z_k)},
\end{align}
by dealing the case $\beta=0$.
\begin{lemma}
We have
\begin{align}
\left.
\frac{\langle \{ m \}_{M,N}|\prod_{j=1}^N \widetilde{\mathcal{B}}(z_j)|\Omega \rangle}
{\prod_{1 \le j < k \le N}(z_j+z_k+2 \beta z_j z_k)}\right|_{\beta=0}
=2^N \prod_{j=1}^N z_j s_\lambda(\{ \bs{z} \}) \label{evaluation}.
\end{align}
\end{lemma}
\begin{proof}
To prove the lemma is equivalent to show
\begin{align}
\langle \{ m \}_{M,N}|\prod_{j=1}^N \widetilde{\mathcal{B}}(z_j)|\Omega \rangle|_{\beta=0}
=2^N \prod_{j=1}^N z_j \prod_{1 \le j < k \le N}(z_j+z_k)
s_\lambda({ \bs{z} })=Q_x({\bs{z} }). \label{whattoprove}
\end{align}
where $Q_x({\bs{z} })$ is the Schur $Q$-function
\begin{align}
Q_x(\bs{z})
=2^N \sum_{w \in S_N}w
\left(\prod_{j=1}^N z_j^{x_j} \prod_{1 \le j < k \le N} 
\frac{z_j+z_k}{z_j-z_k} \right),
\end{align}
where $\bs{z}=\{z_1,\dots,z_N \}$ is a set of variables
and $x$ denotes a strict Young diagram
$x=(x_1,x_2,\dots,x_N)$
with strictly decreasing non-negative integers
$x_1 > x_2 > \cdots > x_N \ge 0$.
\eqref{whattoprove} follows from the fact 
\eqref{skewspecial} derived in the previous section
that the matrix element
of a single $B$-operator at $\beta=0$
is nothing but the skew Schur $Q$-function
\begin{align}
\langle \{ m \}_{M,N+1}|\widetilde{\mathcal{B}}(z)|\{ n \}_{M,N}
\rangle|_{\beta=0}
&=\langle \{ m \}_{M,N+1}|\mathcal{B}(v)|\{ n \}_{M,N} \rangle|_{\beta=0}
\nonumber \\
&=2^{\#(y|x)}
z^{\sum_{j=1}^{N+1} y_j-\sum_{j=1}^N x_j}
\nonumber \\
&=Q_{y/x}(z).
\end{align}
\eqref{whattoprove} follows as a consequence of
the addition formula for the skew Schur $Q$-function
\begin{align}
\langle \{ m \}_{M,N+2}|&\widetilde{\mathcal{B}}(z_1)
\widetilde{\mathcal{B}}(z_2)|\{ n \}_{M,N} \rangle|_{\beta=0} \nonumber \\
&=
\sum_{\ell}
\langle \{ m \}_{M,N+2}|\widetilde{\mathcal{B}}(z_1)|\{ \ell \}_{M,N+1}
\rangle|_{\beta=0}
\langle \{ \ell \}_{M,N+1}|\widetilde{\mathcal{B}}(z_2)|\{ n \}_{M,N}
\rangle|_{\beta=0} \nonumber \\
&=\sum_{w}Q_{y/w}(z_1)Q_{w/x}(z_2) \nonumber \\
&=Q_{y/x}(z_1,z_2).
\end{align}
\end{proof}
Combining Lemma \ref{independence} and \eqref{evaluation},
we have \eqref{equivalentexpressionwavefunction},
and the proof of \eqref{wavefunctiontheorem} is completed.
\end{proof}

\section{Combinatorial formulae for the Schur polynomials}
By combining the analysis of the partition functions
in the previous sections, we obtain combinatorial formulae for
the Schur polynomials.
\begin{theorem}
We have the following combinatorial formula for the Schur polynomials
\begin{align}
s_\lambda(\bs{z})=&
\frac{1}{\prod_{1 \le j < k \le N}(z_j+z_k+2 \beta z_j z_k)}
\sum_{x^{(N)} \succ x^{(N-1)} \succ \dots \succ x^{(0)}=\phi} \prod_{k=1}^N
\Bigg\{z_k^{\sum_{j=1}^k x_j^{(k)}-\sum_{j=1}^{k-1} x_j^{(k-1)}-1} \nonumber \\
\times&\left( \frac{2(1+\beta z_k)}{1+2 \beta z_k} \right)^{\#(x^{(k)}|x^{(k-1)})-1}
\prod_{j=1}^{k-1}
\left(
1+2 \beta z_k(1-\delta_{x_j^{(k-1)} x_{j+1}^{(k)}})
\right)
\Bigg\}, \label{combinatorialformula}
\end{align}
where $\beta$ is an arbitrary parameter.
$x^{(k)}=(x_1^{(k)},\dots,x_k^{(k)})$, $k=0,1,\dots,N$ are strict partitions
satisfying the interlacing relations
$x^{(N)} \succ x^{(N-1)} \succ \dots \succ x^{(0)}=\phi$,
and
$x^{(N)}$ is fixed by the Young diagram
$\lambda=(\lambda_1,\dots,\lambda_N)$ as
$\lambda_j=x_j^{(N)}-N+j-1$.
\end{theorem}
\begin{proof}
We decompose the wavefunction as
\begin{align}
&\bra \{ n \}_{M,N}|\Psi(\{v \}_N) \ket \nn \\
&\quad =\sum_{\{m^{(0)}\},\dots, \{m^{(N-1)}\}}\bra \{ n \}_{M,N}| \prod_{j=1}^N
\left\{\mathcal{B}(v_j)|\{ m^{(N-j)}\}_{M,N-j}\ket \bra \{ m^{(N-j)}\}_{M,N-j}|\right\}|\Omega\ket. \label{decompositionofwavefunction}
\end{align}
We insert the expression
\eqref{wavefunctiontheorem}
in the left hand side of \eqref{decompositionofwavefunction} on one hand.
On the other hand, the right hand side of \eqref{decompositionofwavefunction}
can be expressed using the evaluation of matrix elements
\eqref{tminusoneskew}.
Combining the two expressions and simplifying gives the combinatorial
expression for the Schur polynomials \eqref{combinatorialformula}.
\end{proof}

\begin{example}
Let us check the case $N=2$, $\lambda=(1,0)$.
$x^{(2)}$ is fixed as $x^{(2)}=(1,0)+(2,1)=(3,1)$ and $x^{(0)}=\phi$.
$x^{(1)}$ satisfying the interlacing relation
$x^{(2)} \succ x^{(1)} \succ x^{(0)}$ has three cases
$x^{(1)}=(1)$, $x^{(1)}=(2)$ and $x^{(1)}=(3)$.
Each term in the sum of the right hand side of \eqref{combinatorialformula}
has the contribution
\begin{align}
&z_2^{4-1-1} \left( \frac{2(1+\beta z_2)}{1+2 \beta z_2} \right)^{1-1} z_1^{1-1}, \\
&z_2^{4-2-1} \left( \frac{2(1+\beta z_2)}{1+2 \beta z_2} \right)^{2-1} (1+2 \beta z_2) z_1^{2-1},
\\
&z_2^{4-3-1} \left( \frac{2(1+\beta z_2)}{1+2 \beta z_2} \right)^{1-1} (1+2 \beta  z_2) z_1^{3-1},
\end{align}
which sums up to $(z_1+z_2+2 \beta z_1 z_2)(z_1+z_2)$.
Dividing by $z_1+z_2+2 \beta z_1 z_2$, we have $z_1+z_2$,
which is nothing but $s_{(1,0)}(z_1,z_2)$.
\end{example}

\begin{proposition} The following well-known
identity holds true for the Schur polynomials.
\begin{align}
&\sum_{\lambda \subseteq
(M-N)^N}s_\lambda(\bs{z})
s_{\lambda^\vee}(\bs{y})
=\prod_{1 \le j < k \le N}\frac{1}{(z_k-z_j)(y_j-y_k)}
\mathrm{det}_N \left[ \frac{z_j^M-y_j^M}{z_j -y_k} \right].
\label{cauchyidentity}
\end{align}
\end{proposition}
\begin{proof}
First, substituting the completeness relation, one decomposes 
the scalar product as
\begin{align}
\bra \Psi(\{u\}_N)|\Psi(\{v\}_N)\ket=\sum_{\{ n \}_{M,N}}
\bra \Psi(\{u\}_N)|\{ n \}_{M,N} \ket\bra \{ n \}_{M,N} |\Psi(\{v\}_N)\ket.
\end{align}
Then substituting  the determinant representation for the 
scalar product \eqref{generalscalar} into the RHS of the
above and utilizing the relations \eqref{wavefunctiontheorem}
and \eqref{dualwavefunctiontheorem} yields the Cauchy identity 
\eqref{cauchyidentity} by changing the variables
from $u_j$ and $v_j$ to $y_j=(u_j^{-1}-\beta)^{-1}$ and
$z_j=(v_j^{-1}-\beta)^{-1}$ respectively.
\end{proof}

\section{Enumeration of alternating sign matrices}
In this section, we make an application of the
domain wall boundary partition function
to the enumeration of alternating sign matrices.
See \cite{Bre,Ku,Ku2,Ok2,CP,CP2,BCS,AR,BDZ,Be}
for example of a huge literature on the relation between the
enumeration of alternating sign matrices and integrable vertex models.
The presentation below for the explanation of the relation
between alternating sign matrices and the six-vertex model
follows the lines of \cite{BDZ}.
We take the
homogeneous limit of the domain wall boundary partition function
\eqref{domainwallboundary} or \eqref{dwbpf}
\begin{align}
Z(\{ v \}_n|\{ 1 \})=\langle 1^n | \prod_{k=1}^n \mathcal{B}(v_k)|\Omega \rangle
=\prod_{j=1}^n (2v_j)
\prod_{1 \le j<k \le n}(v_j+v_k). \label{partitionfunctiontoasm}
\end{align}
An algebraic analytic proof of \eqref{domainwallboundary}
is given in  Appendix.
It can also be obtained as a special case of the wavefunction
\eqref{wavefunctiontheorem}.

As a corollary of the domain wall boundary partition function,
we derive a simple expression for a special case of the generating function
of alternating sign matrices.
\begin{definition}
Alternating sign matrices are square matrices
with the following properties: 
\begin{enumerate}
\item each entry is either 0, 1 or $-1$. 
\item there is at least one nonzero entry in each row and column. 
\item  the entries in each row and column sum to 1.
\end{enumerate}
\end{definition}
To satisfy the above conditions, the nonzero entries must 
alternate in sign along  each row and column.
\begin{example}
For $n=3$, there are 7 alternating sign matrices:
\begin{align}
\mathrm{ASM}(3)&=
\left\{
\begin{pmatrix}
1 & 0 & 0 \\
0 & 1 & 0 \\
0 & 0 & 1
\end{pmatrix},
\begin{pmatrix}
0 & 0 & 1 \\
0 & 1 & 0 \\
1 & 0 & 0
\end{pmatrix},
\begin{pmatrix}
1 & 0 & 0 \\
0 & 0 & 1 \\
0 & 1 & 0
\end{pmatrix},
\right. \nn \\
& \qquad \qquad \left.
\begin{pmatrix}
0 & 0 & 1 \\
1 & 0 & 0 \\
0 & 1 & 0
\end{pmatrix},
\begin{pmatrix}
0 & 1 & 0 \\
1 & 0 & 0 \\
0 & 0 & 1
\end{pmatrix},
\begin{pmatrix}
0 & 1 & 0 \\
0 & 0 & 1 \\
1 & 0 & 0
\end{pmatrix},
\begin{pmatrix}
0 & 1 & 0 \\
1 & -1 & 1 \\
0 & 1 & 0
\end{pmatrix}
\right\}.
\label{ASM}
\end{align}
\end{example}
\begin{definition}
Let $\mathrm{ASM}(n)$ be the set consisting of all
$n \times n$ alternating sign matrices.
For $A \in \mathrm{ASM}(n)$,
let us define
\begin{align}
\nu(A):=\sum_{\substack{1 \le i < k \le n \\ 1 \le \ell \le j \le n}}
A_{ij}A_{k \ell},
\end{align}
$\mu(A)$ as the number of $-1$'s in $A$
and $\rho(A)$ as the number of 0's to left of 1 in the first row of $A$.
The generating function of the alternating sign matrices is defined as
\begin{align}
Z_{\mathrm{ASM}}(n,x,y,z)=\sum_{A \in \mathrm{ASM}(n)}x^{\nu(A)}y^{\mu(A)}z^{\rho(A)}.
\end{align}
\end{definition}
Each element of an alternating sign matrix $\mathrm{ASM}(n)$
has one-to-one correspondence with a particular
configuration which makes a non-zero contribution to
the domain wall boundary partition function
of the six-vertex model.
To simplify the explanation,
we identify a particular configuration
of the domain wall boundary partition function
at the coordinate $(j,k)$
with the matrix elements
${}_a \langle *| {}_j \langle * | \mathcal{L}_{a j}(v_k)
|* \rangle_a | * \rangle_j$.
\begin{definition}
Let $\mathrm{6VDW}(n)$ be the set of all particular
configurations making non-zero contributions
to the domain wall boundary partition function
of the six-vertex model.
For $C \in \mathrm{6VDW}(n)$,
Let $\nu(C)$ be the number of vertex configurations
${}_a \langle 1| {}_j \langle 1 | \mathcal{L}_{a j}(v_k)
|1 \rangle_a | 1 \rangle_j$, $j,k=1,\dots,n$ in $C$,
$\mu(C)$ the number of vertex configurations
${}_a \langle 1| {}_j \langle 0 | \mathcal{L}_{a j}(v_k)
|0 \rangle_a | 1 \rangle_j$, $j,k=1,\dots,n$ in $C$,
$\rho(C)$ the number of vertex configurations
${}_a \langle 1| {}_j \langle 1 | \mathcal{L}_{a j}(v_n)
|1 \rangle_a | 1 \rangle_j$, $j=1,\dots,n$ in $C$.
\end{definition}

\begin{theorem} \label{onetoone}
\cite{EKLP}
There is a bijection
between the set of alternating sign matrices
$\{ A \in \mathrm{ASM}(n)  \ | \ \nu(A)=p,
\mu(A)=m, \rho(A)=k \}$
and the set of configurations
making non-zero contributions to the domain wall
boundary partition function of the six-vertex model
$\{ C \in \mathrm{6VDW}(n) \ | \ \nu(C)=p,
\mu(C)=m, \rho(C)=k \}$ by
identifying the matrix elements 1 with
${}_a \langle 0| {}_j \langle 1 | \mathcal{L}_{a j}(v_k)
|1 \rangle_a | 0 \rangle_j$,
$-1$ with ${}_a \langle 1| {}_j \langle 0 | \mathcal{L}_{a j}(v_k)
|0 \rangle_a | 1 \rangle_j$
and
0 with one of the rest of the four non-zero configurations.
\end{theorem}

\begin{proposition} \label{numberoftypes}
\cite{EKLP}
Let $\# {}_a \langle *| {}_j \langle * | \mathcal{L}_{a j}(v_k)
|* \rangle_a | * \rangle_j$ be
the number of vertex configurations \\
${}_a \langle *| {}_j \langle * | \mathcal{L}_{a j}(v_k)
|* \rangle_a | * \rangle_j$, $j,k=1,\dots,n$
in a configuration $C \in \mathrm{6VDW}(n)$.
\begin{align}
\# {}_a \langle 0| {}_j \langle 0 | \mathcal{L}_{a j}(v_k)
|0 \rangle_a | 0 \rangle_j
&=\nu(C), \\
\# {}_a \langle 0| {}_j \langle 1 | \mathcal{L}_{a j}(v_k)
|0 \rangle_a | 1 \rangle_j
&=n(n-1)/2-\nu(C)-\mu(C), \\
\# {}_a \langle 0| {}_j \langle 1 | \mathcal{L}_{a j}(v_k)
|1 \rangle_a | 0 \rangle_j
&=\mu(C)+n, \\
\# {}_a \langle 1| {}_j \langle 0 | \mathcal{L}_{a j}(v_k)
|0 \rangle_a | 1 \rangle_j
&=\mu(C), \\
\# {}_a \langle 1| {}_j \langle 0 | \mathcal{L}_{a j}(v_k)
|1 \rangle_a | 0 \rangle_j
&=n(n-1)/2-\nu(C)-\mu(C), \\
\# {}_a \langle 1| {}_j \langle 1 | \mathcal{L}_{a j}(v_k)
|1 \rangle_a | 1 \rangle_j
&=\nu(C).
\end{align}
\end{proposition}
We use \eqref{partitionfunctiontoasm},
Theorem \ref{onetoone} and Proposition \ref{numberoftypes}
to show the following simple formula for a particular type
of the generating function of alternating sign matrices:
\begin{proposition}
\begin{align}
Z_{\mathrm{ASM}}(n,u-1,u,1)=u^{n(n-1)/2}. \label{asmproposition}
\end{align}
\end{proposition}
\begin{proof}
We first take the homogeneous limit $v_j \to v$, $j=1,\dots,n$
of the partition function
of the six-vertex model \eqref{partitionfunctiontoasm}:
\begin{align}
Z(v):=Z(\{ v \}_n|\{ 1 \})|_{v_1,\dots,v_n \to v}=(2v)^{n(n+1)/2}. \label{comparisonone}
\end{align}
On the other hand, by the definition
of domain wall boundary partition function
and Proposition \ref{numberoftypes}, we have
\begin{align}
Z(v)=\sum_{C \in \mathrm{6VDW}(n)}(a_0 a_1)^{\nu(C)}
(b_0 b_1)^{n(n-1)/2-\nu(C)-\mu(C)} c_0^{\mu(C)}c_1^{\mu(C)+n},
\label{comparisontwo}
\end{align}
where
$a_0=1-\beta v$, $a_1=v$, $b_0=1+\beta v$, $b_1=v$, $c_0=1$, $c_1=2v$. \\
Combining
\eqref{comparisonone} and \eqref{comparisontwo} and simplifying, one gets
\begin{align}
\sum_{C \in \mathrm{6VDW}(n)}
\left(\frac{1-\beta v}{1+\beta v} \right)^{\nu(C)}
\left( \frac{2}{1+\beta v} \right)^{\mu(C)}=
\left(\frac{2}{1+\beta v} \right)^{n(n-1)/2}. \label{asmpropositionpre}
\end{align}
We make change variable from $v$ to $u:=2/(1+\beta v)$.
Finally, we use Theorem \ref{onetoone}
on the one-to-one correspondence
between the set of alternating sign matrices and
the set of configurations of the domain wall boundary partition
function to replace the
sum over $C \in \mathrm{6VDW}(n)$ and the numbers
$\nu(C)$, $\mu(C)$
by the sum over $A \in \mathrm{ASM}(n)$
and the numbers
$\nu(A)$, $\mu(A)$.
Then
\eqref{asmpropositionpre} can be rewritten as \eqref{asmproposition},
which concludes the proof.
\end{proof}

\begin{example}
For $n=3$, there are 7 alternating sign matrices 
as in \eqref{ASM}.
Summing up all the corresponding factors,
we have
$Z_{\mathrm{ASM}}(3,u-1,u,1)=(u-1)^3+2(u-1)^2+(u-1)u+2(u-1)+1=u^3.$
\end{example}
We finally remark that it
may also be possible to derive \eqref{asmproposition}
as a limit of the determinant representation for
$Z_{\mathrm{ASM}}(M,x,y,1)$ in \cite{BDZ}.

\section{Conclusion}
In this paper, we derived a new combinatorial formula
for the Schur polynomials.
The quantum integrability is useful to derive
a formula even for Schur polynomials which is the most
fundamental symmetric polynomials.
The formula expresses Schur polynomials with an additional parameter
besides the spectral parameter.
The other known formula in a similar sense is the Tokuyama formula,
which can be interpreted as a deformation of the Weyl character formula and
the determinant expression.
The representation theoretic meaning of the deformation parameter
in our formula is unknown now which may be worth investigating.
There may be other possibilities to find combinatorial formulae
to express Schur polynomials and other symmetric polynomials in terms
of additional parameters using the power of quantum integrability.
This may be achieved by dealing with partition functions
consisting of different local $L$-operators
or changing global boundary conditions for example.

Besides the traditional problem of the application of partition
functions of the six-vertex model to
the enumeration of alternating sign matrices,
one of the most active line of researches on quantum integrable combinatorics
today is to derive combinatorial formulae
for symmetric polynomials
such as the Cauchy identity, Littlewood identity and so on
by analyzing the transfer matrices or partition functions
of integrable lattice models.
The power of quantum integrable combinatorics is that
one can fuse combinatorics with algebraic analysis
to find and prove identities which seems to be hard
to show in a purely combinatorial or a purely algebraic way.
See \cite{BBF,BMN,MSvertex,BW,BWZ} for finite lattice
and \cite{Bor} for infinite lattice for example
of the recent progresses on this line.

\section*{Acknowledgments}
This work was partially supported by grant-in-aid for
Scientific Research (C) No. 24540393.

\appendix
\def\thesection{\Alph{section}}
\def\reference{\relax\refpar}

\section{Evaluation of the domain wall boundary partition function}
Here we derive the factorized form \eqref{domainwallboundary} of
the domain wall boundary partition function defined by
\begin{align}
Z(\{v\}_N|\{w\}_N):=
\langle 1^N|\prod_{j=1}^N \mathcal{B}_N(v_j)|0^N \rangle,
\label{dwbpf}
\end{align}
where $\mathcal{B}_N(v_j)$ is given by \eqref{B-domain}
and  $\{ w \}_N=\{w_{M+1-N}, \cdots, w_M  \}$.

\begin{proposition} 
The domain wall boundary
partition function \eqref{dwbpf}
is expressed as the following factorized form:
\begin{align}
Z(\{ v \}_N|\{ w \}_N)=\prod_{j=1}^N \frac{2v_j}{w_j^{j}}
\prod_{1 \le j<k \le N}(v_j+v_k). 
\end{align}
\end{proposition}
\begin{proof}
This can be shown in the standard approach due to 
Izergin and Korepin \cite{Kore,Ize}.
First, one shows the following four conditions.
\begin{lemma}
The domain wall boundary partition function
\eqref{dwbpf} satisfies the
following properties. 
\begin{enumerate}
\item $\prod_{j=M+1-N}^M w_j^{N+j-M} Z(\{ v \}_N|\{ w \}_N)$
is symmetric with respect to
the variables $\{w\}_N$. 
\item $Z(\{ v \}_N| \{ w \}_N)/v_N$ 
is a polynomial of degree $N-1$ in $v_N$. 
\item
The following recursive relations between the
domain wall boundary partition functions hold:
\begin{align}
Z(&\{ v \}_N|\{ w \}_N)|_{v_N=\beta^{-1} w_{M-N+1}} \nn \\
&=-\frac{2}{\beta} \prod_{j=M-N+2}^M \left(-\frac{w_{M-N+1}}{\beta w_j}\right)
 \prod_{j=1}^{N-1}\left(1-\frac{\beta v_j}{w_{M-N+1}}\right)
Z(\{ v \}_{N-1}|\{ w \}_{N-1}).
\label{rec-dw}
\end{align}
\item  The case $N=1$ of the domain wall boundary partition function
following form:
\begin{align}
Z(\{ v \}_1|\{ w \}_1)=\frac{2v_1}{w_M}.
\end{align}
\end{enumerate}
\end{lemma}
Properties 2, 3 and 4 can be shown easily with the help of
the graphical representation (see figure~\ref{dw-pic} for Property 3)
of the domain wall boundary partition function 
Let us explain Property 1.
\begin{figure}[ttt]
\begin{center}
\includegraphics[width=0.99\textwidth]{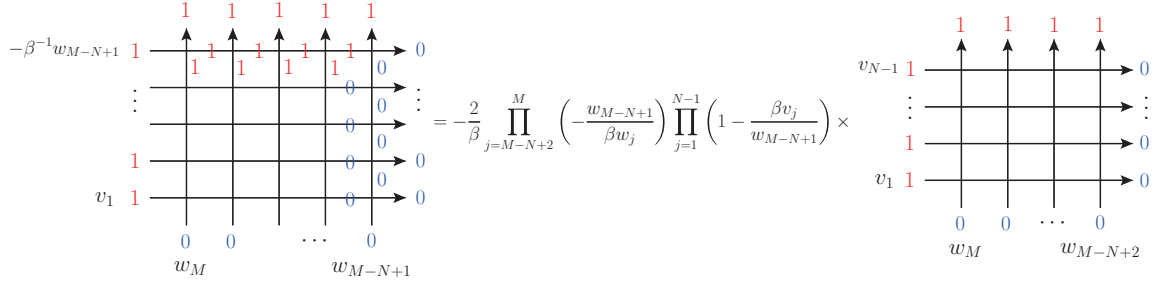}
\end{center}
\caption{A graphical description of \eqref{rec-dw}.}
\label{dw-pic}
\end{figure}
First note that the domain wall boundary partition function
$Z(\{ v \}_N|\{ w \}_N)$
can be re-expressed using the transfer matrix
$T_j(w_{M+1-j})=
\prod_{a=1}^N \mathcal{L}_{a j}(v_{a}/w_{M+1-j})
\in \End( W^{\otimes N} \otimes \mathcal{F}_j)$
propagating in the horizontal direction
as
\begin{align}
Z(\{ v \}_N|\{ w \}_N)={}_a \langle 1^N| C(w_{M+1-N}) \cdots 
C(w_M) 
|0^N \rangle {}_a,
\label{reexp}
\end{align}
where
\begin{align}
T_j(w_{M+1-j})
&=\begin{pmatrix}
A(w_{M+1-j}) & B(w_{M+1-j}) \\
C(w_{M+1-j}) &  D(w_{M+1-j})
\end{pmatrix}_j,
\end{align}
From the $RLL$ relation, one has
\begin{align}
C(w_k)w_j C(w_j)=C(w_j) w_k C(w_k).
\label{deg}
\end{align}
Combining \eqref{reexp} and \eqref{deg} shows Property 1.

The remaining thing to do is to find the explicit forms
of the functions satisfying the properties in the Lemma.
One can easily show that
\begin{align}
Z(\{ v \}_N|\{ w \}_N)=\prod_{j=1}^N (2 v_j) \prod_{j=M+1-N}^M w_j^{M-N-j}
\prod_{1 \le j<k \le N} (v_j+v_k),
\end{align}
satisfies the above four properties.
\end{proof}

\end{document}